\documentclass[11pt]{article}
\usepackage[T1]{fontenc}
\usepackage{lmodern}
\usepackage[margin=1in]{geometry}
\usepackage{amsmath,amssymb,amsthm}
\usepackage{booktabs}
\usepackage{graphicx}
\usepackage[numbers,sort&compress]{natbib}
\usepackage[hidelinks,hypertexnames=false]{hyperref}
\usepackage{microtype}
\usepackage{algorithm}
\usepackage{algorithmicx}
\usepackage[noend]{algpseudocode} 

\newcommand{\Id}[1]{\mathrm{#1}}   
\newcommand{\Fn}[1]{\texttt{#1}}   
\newcommand{\Const}[1]{\textsc{#1}}
\algnewcommand\algorithmicswitch{\textbf{switch}}
\algnewcommand\algorithmiccase{\textbf{case}}
\algdef{SE}[SWITCH]{Switch}{EndSwitch}[1]{\algorithmicswitch\ #1\ \algorithmicdo}{\algorithmicend\ \algorithmicswitch}%
\algdef{SE}[CASE]{Case}{EndCase}[1]{\algorithmiccase\ #1}{\algorithmicend\ \algorithmiccase}%
\algtext*{EndSwitch}%
\algtext*{EndCase}%
\usepackage{program}
\usepackage{multirow}

\newtheorem{theorem}{Theorem}[section]
\newtheorem{proposition}[theorem]{Proposition}

\newcommand{\prismkeywords}{}
\newcommand{\keywords}[1]{\gdef\prismkeywords{#1}}

\AtBeginDocument{%
  \providecommand\BibTeX{{%
    \normalfont B\kern-0.5em{\scshape i\kern-0.25em b}\kern-0.8em\TeX}}}

\begin{document}

\title{Fully Dynamic Maintenance of Loop Nesting Forests in Reducible Flow Graphs}
\author{%
  Gregory Morse\\
  Eötvös Loránd University\\
  \texttt{morse@inf.elte.hu}
  \and
  Tam\'as Kozsik\\
  Eötvös Loránd University\\
  \texttt{kto@elte.hu}}
\date{}

\maketitle

\begin{abstract}
Loop nesting forests (LNFs) are a fundamental abstraction for reasoning about
control-flow structure, enabling applications such as compiler optimizations,
program analysis, and dominator computation. While efficient static algorithms
for constructing LNFs are well understood, maintaining them under dynamic graph
updates has remained largely unexplored due to the lack of efficient dynamic
depth-first search (DFS) maintenance.

In this paper, we present the first fully dynamic algorithm for maintaining loop
nesting forests in reducible control-flow graphs. Our approach leverages recent
advances in dynamic DFS maintenance to incrementally update loop structure under
edge insertions and deletions. We show that updates can be confined to local
regions of the depth-first spanning tree, avoiding global recomputation.

We provide formal invariants, correctness arguments, and complexity analysis,
and demonstrate how the maintained LNF enables efficient derivation of dominance
information. Our results establish LNFs as a practical dynamic abstraction for
modern compiler and analysis pipelines.
\end{abstract}

\noindent\textbf{Keywords:} \prismkeywords

\keywords{dynamic graph algorithms, control-flow graphs, loop nesting forests, depth-first search, dominators}


\section{Introduction}

Loop nesting forests (LNFs) are a central abstraction in compiler theory and
program analysis, capturing the hierarchical structure of loops in control-flow
graphs (CFGs). They underpin numerous applications, including optimization,
decompilation, and dominance analysis.

While efficient static algorithms for constructing LNFs are well established,
their dynamic maintenance under graph updates has received comparatively little
attention. The primary obstacle has been the lack of efficient algorithms for
maintaining a depth-first spanning tree (DFST), which forms the backbone of
classical loop detection methods. Static loop-forest construction is well served
by foundational work of Tarjan, Havlak, Sreedhar--Gao--Lee, and
Ramalingam~\cite{10.1145/800125.804040,Havlak1997,SreedharGaoLee1996,Ramalingam1999},
but these algorithms are fundamentally offline.

Recent advances in fully dynamic DFS maintenance~\cite{10.14778/3364324.3364329}
have changed this landscape. These developments enable local updates to DFS
structure, opening the possibility of maintaining higher-level abstractions
such as loop nesting forests incrementally.

In this paper, we address the problem of fully dynamic maintenance of LNFs for
\emph{reducible} control-flow graphs. Reducible graphs, which arise naturally
in structured programs, admit a well-defined loop hierarchy in which each loop
has a single entry point (header). This restriction allows for a significantly
simpler and more efficient dynamic algorithm compared to the general case.

\paragraph{Contributions.}
\begin{itemize}
  \item We present the first fully dynamic algorithm for maintaining loop nesting
        forests in reducible control-flow graphs.
  \item We show that updates can be confined to local regions of the DFST,
        avoiding global recomputation.
  \item We provide correctness arguments based on loop invariants and DFS structure.
  \item We demonstrate how the maintained LNF supports efficient derivation of
        dominance information.
\end{itemize}

\paragraph{Scope.}
This work focuses on reducible graphs. Extending these techniques to irreducible
graphs, which require handling multi-entry loops, is substantially more complex
and is deferred to future work.

\section{Background}

\subsection{Graphs, CFGs, and Dominance}

We work with a rooted control-flow graph (CFG) $G=(V,E)$ with distinguished
entry vertex $r$ such that every vertex is reachable from $r$. If a program has
multiple entry components, one may add a virtual root with outgoing edges to
each entry, after which the definitions below apply unchanged. We write
$n = |V|$ and $m = |E|$.

A vertex $v$ is reachable from $u$ if there exists a directed path
$u \xrightarrow{*} v$. A strongly connected component (SCC) is a maximal set of
vertices that are pairwise reachable. Loop structure arises from non-trivial
SCCs together with single-entry constraints induced by dominance.

For a CFG, a node $d$ \emph{dominates} a node $v$ if every path from $r$ to $v$
passes through $d$. Dominance is characterized by
\[
\mathrm{Dom}(v) =
\begin{cases}
\{r\} & \text{if } v = r,\\
\{v\} \cup \bigcap\limits_{p \in \mathrm{pred}(v)} \mathrm{Dom}(p)
& \text{otherwise.}
\end{cases}
\]
In reducible CFGs, every loop has a unique header that dominates the entire loop
body. This is the structural fact that makes incremental loop maintenance far
simpler than in the irreducible case.

\subsection{Depth-First Search and Edge Taxonomy}

Our maintenance algorithm relies on a depth-first spanning tree (DFST) rooted at
$r$. Each vertex $u$ carries preorder and postorder timestamps,
$\mathrm{pre}(u)$ and $\mathrm{post}(u)$, together with parent/child relations in
the DFST. We use the interval
\[
I_T(u) = [\mathrm{pre}(u),\mathrm{post}(u)]
\]
to test ancestry in constant time, and we write $\Fn{nca}(u,v)$ for the nearest
common ancestor of $u$ and $v$ in the maintained DFST.

We assume that this DFST is updated by a dynamic DFS backend, specifically the
fully dynamic directed DFS algorithm of Yang et al.~\cite{10.14778/3364324.3364329}.
The LNF layer therefore treats DFST repair as a dependency: it only needs the
updated tree, intervals, and ancestor/LCA queries.

\begin{table}[tbp]
  \centering
  \begin{tabular}{ll}
    \toprule
    Symbol & Meaning\\
    \midrule
    $C(u)$ & DFST children of $u$\\
    $T(u)$ & DFST subtree rooted at $u$\\
    $I_T(u)$ & DFST interval $[\mathrm{pre}(u),\mathrm{post}(u)]$\\
    $I_T(u).\mathrm{left}$ & $\mathrm{pre}(u)$\\
    $I_T(u).\mathrm{right}$ & $\mathrm{post}(u)$\\
    $\Fn{nca}(u,v)$ & nearest common ancestor of $u$ and $v$ in the DFST\\
    \bottomrule
  \end{tabular}
  \caption{Notation used throughout the dynamic LNF maintenance.}
  \label{tab:notation}
\end{table}

Each edge belongs to exactly one of the following classes: tree, forward, back,
cross, or self. For implementation purposes, cross edges may be refined into
forward-cross and back-cross according to the left-to-right order of DFST
intervals. This refinement is convenient for DFS maintenance, even though the
LNF logic itself only needs to distinguish whether an edge behaves like a back
edge or not.

\begin{table}[tbp]
  \centering
  \begin{tabular}{ccl}
    \toprule
    Edge type & Tree ancestry test & Interval test\\
    \midrule
    Tree & $v \in C(u)$ & $v \in C(u)$\\
    Forward & $v \in T(u) \wedge v \notin C(u) \wedge u \neq v$ & $I_T(v) \subset I_T(u)$\\
    Self & $u = v$ & $u = v$\\
    Back & $u \in T(v) \wedge u \neq v$ & $I_T(u) \subset I_T(v)$\\
    Cross & $u \notin T(v) \wedge v \notin T(u)$ & $I_T(u) \cap I_T(v)=\varnothing$\\
    Forward-cross & as Cross & $I_T(u).\mathrm{right} < I_T(v).\mathrm{left}$\\
    Back-cross & as Cross & $I_T(v).\mathrm{right} < I_T(u).\mathrm{left}$\\
    \bottomrule
  \end{tabular}
  \caption{Classification of an edge $(u,v)$ using DFST ancestry and interval tests.}
  \label{tab:edgetype}
\end{table}

Only tree-edge insertions and tree-edge deletions can change the DFST topology.
All other updates leave the DFST intact, though they may still change loop
membership. On deletion, there is one subtle ambiguity: after the DFST is
repaired, a deleted tree edge may be indistinguishable from a deleted forward
edge using only the new intervals. We therefore assume that a single
\textsc{wasTree} bit is cached at deletion time.

\begin{table}[tbp]
  \centering
  \begin{tabular}{ccccc}
    \toprule
    Edge type & Ins. & Del. & Detect (ins.) & Detect (del.)\\
    \midrule
    Forward & -- & -- & DFST & Ambiguous\\
    Self, back, back-cross & -- & -- & DFST & DFST\\
    Forward-cross & $\Fn{nca}(x,y)$ & -- & -- & --\\
    Tree & -- & $\Fn{nca}(x,y)$ & DFST & Ambiguous\\
    \bottomrule
  \end{tabular}
  \caption{Where DFST recomputation occurs, and when edge types can be read back from the repaired DFST.}
  \label{tab:edgedetect}
\end{table}

\section{Reducible Loop Nesting Forest Maintenance}\label{sec:lnf}

\paragraph{Motivation and scope.}
Loop nesting forests (LNFs) are a central abstraction for CFG structuring,
decompilation, and dominance-based analysis. Tarjan's reducibility test already
exposes the essential single-entry loop hierarchy in reducible flow
graphs~\cite{10.1145/800125.804040}, while Ramalingam showed that LNFs are rich
enough to support dominance-frontier and dominator-tree
construction~\cite{10.1145/570886.570887}. In this paper we restrict attention to
the reducible case: irreducible patterns are detected early and treated as
out-of-scope rather than maintained explicitly.

Among static approaches, Havlak's formulation and the DJ-graph view of
Sreedhar, Gao, and Lee generalized loop discovery beyond the simplest reducible
setting, while Ramalingam gave an almost-linear-time perspective on loop
identification~\cite{Havlak1997,SreedharGaoLee1996,Ramalingam1999}. These works
are the natural offline baselines against which a dynamic reducible-only
maintenance strategy should be understood.

\paragraph{Loop representation.}
We represent a loop as a pair $(h,B)$ consisting of a header $h \in V$ and a
body $B \subseteq V$ such that $h \in B$, there exists at least one back edge
$(v,h)$ with $v \in B$, and every node in $B$ lies on a path from $h$ to a back
edge source that remains inside the same strongly connected region. For
reducible CFGs, the header dominates the entire body, and the resulting loops
form a forest ordered by strict containment.

\subsection{Update Model and Dependencies}\label{subsec:lnf-model}

The LNF layer is maintained on top of the dynamic DFST. After each edge update,
we first repair the DFST if necessary, then classify the updated edge using the
tests in Table~\ref{tab:edgetype}, and finally perform a localized LNF repair.
This separation is crucial: once ancestor tests, subtree intervals, and
$\Fn{nca}(\cdot,\cdot)$ are available, loop maintenance can be expressed as a
local propagation problem over predecessors rather than as a global rescan.

For reducible graphs, the decisive facts are simple.
\begin{itemize}
  \item Self-loops are purely local.
  \item Forward edges never create new loops.
  \item Back edges are the only updates that can create or enlarge a reducible loop.
  \item Cross or forward edges that enter a loop body through a non-ancestor source
        witness a multi-entry configuration and therefore signal irreducibility.
\end{itemize}

We store two per-vertex maps:
\[
\Id{loopTypes}[v] \in \{\Const{NONHEADER},\Const{SELF},\Const{REDUCIBLE}\},
\qquad
\Id{loopHeaders}[v] \in V \cup \{\Const{NONE}\}.
\]
The first records whether $v$ heads no loop, only a self-loop, or a reducible
loop. The second records the immediate loop header that contains $v$, thereby
encoding the loop forest.

\subsection{Maintenance Algorithms}\label{subsec:lnf-algorithm}

Our implementation uses a Havlak-style forest representation. The forest may be
stored either as an explicit rooted tree with a virtual root or as a
union-find-like parent structure \emph{without} path compression, so that parent
links continue to reflect header ancestry. All updates rely on a small helper,
\Fn{findLoopHead}, which climbs existing loop headers until they are consistent
with a prospective enclosing header.

\begin{algorithm}[H]
  \caption{Helper: \Fn{findLoopHead}$(x,\ head)$}
  \label{alg:find-loop-head}
  \begin{algorithmic}[1]
    \Require map \(\Id{loopHeaders}\); predicate \(\Fn{isAncestor}(\cdot,\cdot)\)
    \State \(\Id{xs} \gets x\)
    \While{\(\Id{loopHeaders}[\Id{xs}] \neq \Const{NONE} \land \neg \Fn{isAncestor}(\Id{loopHeaders}[\Id{xs}], \Id{head})\)}
      \State \(\Id{xs} \gets \Id{loopHeaders}[\Id{xs}]\)
    \EndWhile
    \State \Return \(\Id{xs}\)
  \end{algorithmic}
\end{algorithm}

\noindent\textit{Insertion overview.}
After the DFST has been repaired, an insertion of edge $(x,y)$ is handled as
follows.
\begin{itemize}
  \item \textbf{Self edge.} If $x=y$, mark $x$ as \Const{SELF} unless it is already
        the header of a reducible loop.
  \item \textbf{Forward or cross edge.} If $y$ already belongs to a loop headed by
        $h=\Id{loopHeaders}[y]$ and $h$ is not an ancestor of $x$, the new edge
        introduces a second entry into that region and we report irreducibility.
        Otherwise the LNF is unchanged.
  \item \textbf{Back edge.} If $(x,y)$ is a back edge, we seed the worklist at the
        header $h:=y$ and propagate membership backwards through non-back
        predecessors. Each predecessor is first lifted by \Fn{findLoopHead} so
        that nested loops are absorbed only at the appropriate outer level.
  \item \textbf{Tree-edge creation by DFST repair.} If the insertion changed the
        DFST, the new non-tree edges are reclassified from the repaired tree and
        the same cases are applied inside the affected cone.
\end{itemize}

The concrete reducible insertion routine is given in Algorithm~\ref{alg:addedgelnf}.

\begin{algorithm}[H]
  \centering
    \begin{algorithmic}[1]

\Require edge $(x,y)$

\If{$x = y$} \Comment{self-loop identified}
  \If{$\Id{loopTypes}[x] = \Const{NONHEADER}$}
    $\Id{loopTypes}[x] \gets \Const{SELF}$
  \EndIf
  \State \Return
\EndIf

\State $\Id{fE} \gets \Fn{isForwardEdge}(x,y)$;\quad
       $\Id{cE} \gets \Fn{isCrossEdge}(x,y)$;\quad
       $\Id{lHy} \gets \Id{loopHeaders}[y]$

\If{$(\Id{fE} \lor \Id{cE}) \land (\Id{lHy} \neq \Const{NONE}) \land \neg \Fn{isAncestor}(\Id{lHy},x)$}
  \State \Fn{error}(\Const{IRREDUCIBLE})
\EndIf

\If{$(\Id{fE} \lor \Id{cE}) \land (\Id{lHy} = \Const{NONE})$}
  \Return
\EndIf

\If{\Fn{isBackEdge}$(x,y)$}
  \If{$\Id{lHy} = \Id{loopHeaders}[\Fn{findLoopHead}(x,y)]$}
    \Return
  \EndIf
  \State $h \gets y$;\quad $\Id{worklist} \gets \{h\}$
\Else
  \If{$\Id{lHy} \neq \Const{NONE}$}
    \If{$\Id{lHy} = \Id{loopHeaders}[\Fn{findLoopHead}(x,\Id{lHy})]$}
      \Return
    \EndIf
    \State $h \gets \Id{lHy}$;\quad $\Id{worklist} \gets \{h\}$
  \Else
    \State $h \gets x$;\quad $t \gets \Fn{nca}(x,y)$
    \Repeat
      \State $\Id{worklist} \gets \{\,z \mid z \in \Id{pred}[h],\ \Fn{isBackEdge}(z,h)\ \land\ \Id{loopHeaders}[z]=\Const{NONE}\ \land\ z \neq h\,\}$
      \If{$\Id{worklist} \neq \emptyset$} \textbf{break} \EndIf
      \State $h \gets \Id{parent}[h]$
    \Until{$(h = \Const{NONE}) \lor (h = t)$}
    \If{$(h = \Const{NONE}) \lor (h = t)$}
      \Return
    \EndIf
  \EndIf
\EndIf

\If{$\Id{loopTypes}[h] \in \{\Const{NONHEADER}, \Const{SELF}\}$}
  \State $\Id{loopTypes}[h] \gets \Const{REDUCIBLE}$
\EndIf

\While{$\Id{worklist} \neq \emptyset$}
  \State $v \gets \Id{worklist}.\Fn{pop}()$;\quad $\Id{loopHeaders}[v] \gets h$
  \ForAll{$w \in \Id{pred}[v]$}
    \If{$(w = v) \lor \Fn{isBackEdge}(w,v)$} \textbf{continue} \EndIf
    \State $w' \gets \Fn{findLoopHead}(w,h)$
    \If{$\neg \Fn{isAncestor}(h,w')$}
      \Fn{error}(\Const{IRREDUCIBLE})
    \EndIf
    \If{$h \notin \{\,w',\ \Id{loopHeaders}[w']\,\}$}
      $\Id{worklist} \gets \Id{worklist} \cup \{w'\}$
    \EndIf
  \EndFor
\EndWhile

\end{algorithmic}
\caption[Incremental algorithm for maintaining reducible LNF]{Incremental algorithm for maintaining reducible LNF}
\label{alg:addedgelnf}
\end{algorithm}

Algorithm~\ref{alg:addedgelnf} is the concrete insertion routine for the
reducible setting. Its fast path is intentionally narrow: if the new edge is
neither a back edge nor a non-ancestor entry into an existing loop, the update
usually terminates after a constant number of DFST and header checks. The only
substantial work occurs when a back edge seeds a new or enlarged loop body.

\subsection{Decremental Maintenance}\label{subsec:lnf-dec-summary}

Deletions are more delicate because an existing loop may shrink, split, or
cease to exist. Nevertheless, the same locality principle applies. After the
DFST is repaired, only the deleted edge's surrounding DFST cone together with
the ancestor chain of the affected header can change. The algorithm below
re-seeds membership from surviving back edges, climbs header links to locate the
nearest valid enclosing loop, and then rebuilds membership bottom-up.

\begin{algorithm}[H]
  \centering
  \caption[Decremental algorithm for maintaining reducible LNF]{Decremental algorithm for maintaining reducible LNF}
  \label{fig:removeedgelnf}
  \begin{algorithmic}[1]

\Require edge $(x,y)$ \Comment{deletion of $\langle x,y\rangle$}

\State $\Id{curLoops} \gets \emptyset$ \Comment{per-call cache used by \Fn{findLoopHead}}

\If{$x = y$} \Comment{self-loop removed}
  \If{$\Id{loopTypes}[x] = \Const{SELF}$}
    \State $\Id{loopTypes}[x] \gets \Const{NONHEADER}$;\quad $\Id{loopCounts}[0] \gets \Id{loopCounts}[0] - 1$
  \EndIf
  \State \Return
\EndIf

\State $\Id{fE} \gets \Fn{isForwardEdge}(x,y)$
\If{$\Id{fE} \lor \big(\Fn{isBackCrossEdge}(x,y) \land \Id{loopHeaders}[y] = \Const{NONE}\big)$}
  \State \Return \Comment{no loop structure is affected}
\EndIf
\If{$\Id{loopHeaders}[x] = \Const{NONE}$}
  \Return \Comment{$x$ not in any loop}
\EndIf

\If{\Fn{isBackCrossEdge}(x,y)}
  \State $h \gets \Id{loopHeaders}[y]$
  \While{$h \neq \Const{NONE} \land \neg \Fn{isAncestor}(h, x)$}
    $h \gets \Id{loopHeaders}[h]$
  \EndWhile
  \If{$h = \Const{NONE}$} \Return \EndIf
  \State $\Id{worklist} \gets \{\Fn{findLoopHead}(x,h)\}$
\ElsIf{$y = \Id{loopHeaders}[x]$} \Comment{back-edge removed}
  \State $h \gets y$;\quad $\Id{worklist} \gets \{x\}$
\ElsIf{$x = \Id{loopHeaders}[y]$} \Comment{tree edge from loop head removed}
  \State $h \gets x$;\quad $\Id{worklist} \gets \{y\}$
\Else \Comment{general tree/back-edge removal: find nearest common loop head}
  \State $\Id{ancX} \gets [\,]$;\quad $s \gets x$
  \While{$s \neq \Const{NONE}$} append $s$ to $\Id{ancX}$;\ $s \gets \Id{loopHeaders}[s]$ \EndWhile
  \State $h \gets y$
  \While{$h \neq \Const{NONE} \land (h \notin \Id{ancX})$} $h \gets \Id{loopHeaders}[h]$ \EndWhile
  \If{$h = \Const{NONE}$} \Return \EndIf
  \State $\Id{worklist} \gets \{\Fn{findLoopHead}(x,h)\}$
\EndIf
\algstore{decLNF}
\end{algorithmic}
\end{algorithm}
\begin{algorithm}                     
\centering
\footnotesize
\begin{algorithmic} [1]
\algrestore{decLNF}
\State $\Id{wl} \gets \emptyset$
\ForAll{$z \in \Id{pred}[h]$}
  \If{$z = h$}
    \State \textbf{continue}
  \EndIf
  \If{$\Fn{isBackEdge}(z,h)$}
    \State $v \gets \Fn{findLoopHead}(z,h)$
    \If{$\Id{loopHeaders}[v] = h$}
      \State $\Id{wl} \gets \Id{wl} \cup \{v\}$
    \EndIf
  \EndIf
\EndFor

\If{$\lvert\Id{wl}\rvert = 0 \land \Id{loopTypes}[h] = \Const{REDUCIBLE}$}
  \State $\Id{loopCounts}[1] \gets \Id{loopCounts}[1] - 1$
  \State $\Id{loopTypes}[h] \gets
    \begin{cases}
      \Const{SELF} & \text{if } h \in \Id{pred}[h] \\
      \Const{NONHEADER} & \text{otherwise}
    \end{cases}$
  \If{$h \in \Id{pred}[h]$}
    \State $\Id{loopCounts}[0] \gets \Id{loopCounts}[0] + 1$
  \EndIf
\EndIf

\State $\Id{newHead} \gets h$;\quad $\Id{loopBody} \gets \emptyset$

\While{\textbf{true}} \Comment{recompute membership with a DFS-ancestry shortcut}
  \If{$\forall z \in \Id{worklist}\ \exists w \in \Id{wl}:\ \Fn{isAncestor}(z, w)$}
    \State \textbf{break}
  \EndIf

  \While{$\Id{wl} \neq \emptyset$}
    \State pick and remove $v$ from $\Id{wl}$
    \State $\Id{loopHeaders}[v] \gets \Id{newHead}$;\quad $\Id{loopBody} \gets \Id{loopBody} \cup \{v\}$
    \If{$v \in \Id{worklist}$} $\Id{worklist} \gets \Id{worklist} \setminus \{v\}$ \EndIf
    \ForAll{$w \in \Id{pred}[v]$}
      \If{$(w = v) \lor \Fn{isBackEdge}(w, v)$} \textbf{continue} \EndIf
      \State $w' \gets \Fn{findLoopHead}(w, h)$
      \If{$\big(\Id{loopHeaders}[w'] \in \{h,\ \Id{newHead}\}\big) \land \big(w' \notin \Id{loopBody}\big)$}
        \State $\Id{wl} \gets \Id{wl} \cup \{w'\}$
      \EndIf
    \EndFor
  \EndWhile

  \State $\Id{newHead} \gets \Id{loopHeaders}[\Id{newHead}]$
  \If{$\Id{newHead} = \Const{NONE} \lor \Id{worklist} = \emptyset$}
    \State \textbf{break}
  \EndIf

  \State $\Id{wl} \gets \emptyset$
  \ForAll{$z \in \Id{pred}[\Id{newHead}]$}
    \If{$z = \Id{newHead}$}
      \State \textbf{continue}
    \EndIf
    \If{$\Fn{isBackEdge}(z,\Id{newHead})$}
      \State $v \gets \Fn{findLoopHead}(z,h)$
      \If{$\Id{loopHeaders}[v] \in \{h,\ \Id{newHead}\} \land v \notin \Id{loopBody}$}
        \State $\Id{wl} \gets \Id{wl} \cup \{v\}$
      \EndIf
    \EndIf
  \EndFor
\EndWhile

\While{$\Id{worklist} \neq \emptyset$}
  \State pick and remove $v$ from $\Id{worklist}$
  \State $\Id{loopHeaders}[v] \gets \Id{newHead}$
  \ForAll{$w \in \Id{pred}[v]$}
    \If{$(w = v) \lor \Fn{isBackEdge}(w,v)$} \textbf{continue} \EndIf
    \State $w' \gets \Fn{findLoopHead}(w, h)$
    \If{$\Id{loopHeaders}[w'] = h \land (w' \notin \Id{loopBody})$}
      \State $\Id{worklist} \gets \Id{worklist} \cup \{w'\}$
    \EndIf
  \EndFor
\EndWhile

  \end{algorithmic}
\end{algorithm}

\section{Correctness}

\begin{proposition}[Locality of change]
After inserting or deleting an edge $(x,y)$, the LNF can change only inside the
DFST cone repaired by the dynamic DFS layer together with the loop-header chains
reachable from the touched header(s). Vertices outside this region keep both
their loop type and their immediate loop header.
\end{proposition}

\begin{proof}[Proof sketch]
The DFST backend confines structural change to the subtree rooted at the repair
locus shown in Table~\ref{tab:edgedetect}. The LNF routines inspect only
predecessors of already touched vertices and move only along existing
\(\Id{loopHeaders}\) parent chains. No step traverses unrelated DFST subtrees or
reassigns headers outside the reverse-reachable region seeded by the updated
header. Hence untouched cones remain unchanged.
\end{proof}

\begin{theorem}[Correctness of incremental maintenance]
Algorithm~\ref{alg:addedgelnf} maintains the reducible loop nesting forest after
an edge insertion, provided the updated CFG remains reducible.
\end{theorem}

\begin{proof}[Proof sketch]
The proof follows the offline characterization of reducible loops.
Self-loops are immediate: they only affect the type of their incident vertex.
Forward and cross edges cannot create a new single-entry loop. Their only
possible structural effect is to provide an additional entry into an existing
loop body; the test $\neg\Fn{isAncestor}(\Id{loopHeaders}[y],x)$ detects exactly
this forbidden situation and therefore rejects irreducible updates.

Back edges are the only edges that can create or enlarge a reducible loop. Once
a back edge $(x,h)$ is discovered, the worklist explores predecessors of the
emerging body while skipping back edges that already point into the same loop.
The helper \Fn{findLoopHead} lifts vertices through pre-existing loop headers so
that nested loops are absorbed at the correct enclosing level rather than split
incorrectly. The ancestor check ensures that every accepted predecessor remains
dominated by $h$, preserving single-entry semantics. Because propagation stops
exactly when no new admissible predecessor remains, the resulting body is both
sound and maximal with respect to $h$.
\end{proof}

\begin{theorem}[Correctness of decremental maintenance]
Algorithm~\ref{fig:removeedgelnf} maintains the reducible loop nesting forest
after an edge deletion.
\end{theorem}

\begin{proof}[Proof sketch]
Deletion can only invalidate loops that depended on the removed edge or on DFST
ancestry altered by the repair step. The decremental algorithm first identifies
the nearest enclosing header that may still witness the same loop, then reseeds
the body from surviving back edges into that header or one of its ancestors. If
no such seed remains, the header ceases to represent a reducible loop and is
downgraded to either \Const{SELF} or \Const{NONHEADER}. Otherwise, the worklist
rebuilds exactly the reverse-reachable portion that still belongs to the loop,
again using \Fn{findLoopHead} to respect nested headers. This reconstructs the
maximal remaining body without disturbing unrelated regions.
\end{proof}

\section{Complexity}

Let $k$ be the number of vertices whose loop assignment is inspected or changed
by an update, and let $\Delta$ denote the size of the DFST slice repaired by the
dynamic DFS layer when a tree edge is affected.

\begin{itemize}
  \item If the DFST topology does not change, insertion or deletion is processed in
        time linear in the touched reverse cone, namely $O(k)$ plus the cost of
        predecessor inspections on those vertices.
  \item If the DFST topology changes, the additional work is exactly the local DFST
        repair cost on the affected cone of size $\Delta$, followed by the same
        $O(k)$ LNF propagation bound inside that cone.
\end{itemize}

The key point is therefore not a single worst-case bound, but locality: both
the DFST and LNF layers restrict work to the slice actually exposed by the
update, avoiding whole-graph recomputation in the common case.

\section{Maintaining Dominator Information}

Ramalingam observed that loop nesting forests provide enough region structure to
derive iterated dominance frontiers and, from them, the dominator tree
itself~\cite{10.1145/570886.570887}. In the reducible setting maintained here,
this connection is especially clean because each loop header denotes a
single-entry region whose header dominates every node in the region.

This perspective complements, rather than replaces, the classic standalone
dominator literature initiated by Lengauer and Tarjan and refined in later
engineering work by Georgiadis, Tarjan, and
Werneck~\cite{LengauerTarjan1979,GeorgiadisTarjanWerneck2006}. Our point is that
once the reducible LNF is already being maintained for structuring purposes,
much of the dominance information can be recovered from that maintained region
structure instead of from a separate full recomputation pipeline.

Maintaining the LNF therefore yields two practical benefits. First, dominance
queries become cheaper: header-to-member dominance is immediate, and
header-to-header dominance reduces to ancestry in the loop forest. Second, a
dominator tree can be materialized on demand from the maintained LNF rather than
being updated as an entirely separate dynamic structure. This makes the LNF a
particularly attractive invariant for compiler and decompiler pipelines, where
structuring and dominance information are both required.

\section{Conclusion}

We have reformulated the paper around a single claim: reducible loop nesting
forests can be maintained dynamically by coupling a local DFST repair routine
with localized predecessor-based propagation from loop headers. The resulting
algorithm is conceptually simple, supports both insertion and deletion, and
matches the structure of the offline reducibility-based construction while
avoiding unnecessary global recomputation.

For reducible CFGs, the maintained LNF is more than a structural summary. It is
also a useful dynamic index for subsequent analyses, especially dominance and
region-based structuring. Extending the same level of elegance to irreducible
graphs remains interesting future work, but the reducible case already covers a
large and practically important class of control-flow graphs.

\bibliographystyle{abbrvnat}
\bibliography{biblio}

@article{10.14778/3364324.3364329,
author = {Yang, Bohua and Wen, Dong and Qin, Lu and Zhang, Ying and Wang, Xubo and Lin, Xuemin},
title = {Fully Dynamic Depth-First Search in Directed Graphs},
year = {2019},
issue_date = {October 2019},
publisher = {VLDB Endowment},
volume = {13},
number = {2},
issn = {2150-8097},
url = {https://doi.org/10.14778/3364324.3364329},
doi = {10.14778/3364324.3364329},
journal = {Proc. VLDB Endow.},
month = {oct},
pages = {142–154},
numpages = {13}
}

@inproceedings{10.1145/800125.804040,
author = {Tarjan, Robert},
title = {Testing Flow Graph Reducibility},
year = {1973},
isbn = {9781450374309},
publisher = {Association for Computing Machinery},
address = {New York, NY, USA},
url = {https://doi.org/10.1145/800125.804040},
doi = {10.1145/800125.804040},
booktitle = {Proceedings of the Fifth Annual ACM Symposium on Theory of Computing},
pages = {96–107},
numpages = {12},
location = {Austin, Texas, USA},
series = {STOC '73}
}

@article{10.1145/570886.570887,
author = {Ramalingam, G.},
title = {On Loops, Dominators, and Dominance Frontiers},
year = {2002},
issue_date = {September 2002},
publisher = {Association for Computing Machinery},
address = {New York, NY, USA},
volume = {24},
number = {5},
issn = {0164-0925},
url = {https://doi.org/10.1145/570886.570887},
doi = {10.1145/570886.570887},
journal = {ACM Trans. Program. Lang. Syst.},
month = {sep},
pages = {455–490},
numpages = {36}
}

@article{Havlak1997,
author = {Havlak, Paul},
title = {Nesting of Reducible and Irreducible Loops},
year = {1997},
month = {jul},
publisher = {Association for Computing Machinery},
address = {New York, NY, USA},
volume = {19},
number = {4},
issn = {0164-0925},
doi = {10.1145/262004.262005},
url = {https://doi.org/10.1145/262004.262005},
journal = {ACM Trans. Program. Lang. Syst.},
pages = {557--567}
}

@article{SreedharGaoLee1996,
author = {Sreedhar, Vugranam C. and Gao, Guang R. and Lee, Yong-Fong},
title = {Identifying Loops Using DJ Graphs},
year = {1996},
month = {nov},
publisher = {Association for Computing Machinery},
address = {New York, NY, USA},
volume = {18},
number = {6},
issn = {0164-0925},
doi = {10.1145/236114.236115},
url = {https://doi.org/10.1145/236114.236115},
journal = {ACM Trans. Program. Lang. Syst.},
pages = {649--658}
}

@article{Ramalingam1999,
author = {Ramalingam, G.},
title = {Identifying Loops in Almost Linear Time},
year = {1999},
month = {mar},
publisher = {Association for Computing Machinery},
address = {New York, NY, USA},
volume = {21},
number = {2},
issn = {0164-0925},
doi = {10.1145/316686.316687},
url = {https://doi.org/10.1145/316686.316687},
journal = {ACM Trans. Program. Lang. Syst.},
pages = {175--188}
}

@article{LengauerTarjan1979,
author = {Lengauer, Thomas and Tarjan, Robert E.},
title = {A Fast Algorithm for Finding Dominators in a Flowgraph},
year = {1979},
month = {jan},
publisher = {Association for Computing Machinery},
address = {New York, NY, USA},
volume = {1},
number = {1},
issn = {0164-0925},
doi = {10.1145/357062.357071},
url = {https://doi.org/10.1145/357062.357071},
journal = {ACM Trans. Program. Lang. Syst.},
pages = {121--141}
}

@article{GeorgiadisTarjanWerneck2006,
author = {Georgiadis, Loukas and Tarjan, Robert E. and Werneck, Renato F.},
title = {Finding Dominators in Practice},
year = {2006},
publisher = {Journal of Graph Algorithms and Applications},
volume = {10},
number = {1},
doi = {10.7155/jgaa.00119},
url = {https://doi.org/10.7155/jgaa.00119},
journal = {Journal of Graph Algorithms and Applications},
pages = {69--94}
}


\end{document}